\documentclass[11pt]{amsart}
\usepackage{amssymb,color}
\numberwithin{equation}{section}

\newcommand{\bea}{\begin{eqnarray}}
\newcommand{\eea}{\end{eqnarray}}
\newcommand{\ba}{\begin{array}}
\newcommand{\ea}{\end{array}}
\newcommand{\edc}{\end{document}}
\newcommand{\bc}{\begin{center}}
\newcommand{\ec}{\end{center}}
\newcommand{\be}{\begin{equation}}
\newcommand{\ee}{\end{equation}}

\newcommand{\dsf}{\displaystyle\frac}

\def\cb{{\mathcal B}}

\def\cf{{\mathcal F}}
\def\cg{{\mathcal G}}

\def\gh{{\frak h}}

\def\bc{{\mathbb C}}

\def\bn{{\mathbb N}}

\def\bq{{\mathbb Q}}

\def\bz{{\mathbb Z}}

\def\a{\alpha}

\def\g{\gamma}  
\def\d{\delta}  \def\D{\Delta}

\def\e{\epsilon}
\def\l{\lambda} 
\def\k{\kappa}
\def\m{\mu}

\def\s{\sigma} 
\def\t{\theta}

\def\w{\omega} \def\Om{\Omega}

\def\h{{\mathbf{h}}}

\def\xb{{\mathbf{x}}}
\def\sb{{\mathbf{s}}}
\def\yb{{\mathbf{y}}}
\def\Bb{{\mathbf{B}}}

\newtheorem{thm}{Theorem}[section]
\newtheorem{lem}[thm]{Lemma}
\newtheorem{cor}[thm]{Corollary}

\theoremstyle{remark}
\newtheorem{rem}{Remark}[section]
\newtheorem{ex}{Example}[section]

\begin{document}

\title[phase transition for $p$-adic Potts model]
{On phase transition for one dimensional countable state $P$-adic
Potts model}


\author{Farrukh Mukhamedov}
\address{Farrukh Mukhamedov\\
 Department of Computational \& Theoretical Sciences\\
Faculty of Science, International Islamic University Malaysia\\
P.O. Box, 141, 25710, Kuantan\\
Pahang, Malaysia} \email{{\tt far75m@yandex.ru}, {\tt
farrukh\_m@iiu.edu.my}}

\begin{abstract}
In the present paper we shall consider countable state $p$-adic
Potts model on $\bz_+$. A main aim is to establish the existence of
the phase transition for the model. In our study, we essentially use
one dimensionality of the model. To show it we reduce the problem,
to the investigation of an infinite-dimensional nonlinear equation.
We find a condition on weights to show that the derived equation has
two solutions, which yields the existence of the phase transition.
We prove that measures corresponding to  first and second solutions
are a $p$-adic Gibbs and generalized $p$-adic Gibbs measures,
respectively. Note that it turns out that the finding condition does
not depend on values of the prime $p$, and therefore, an analogous
fact is not true when the number of spins is finite. Note that, in
the usual real case, if one considers one dimensional
translation-invariant model with nearest neighbor interaction, then
such a model does not exhibit a phase transition. Nevertheless, we
should stress that in our model there does not occur the strong
phase transition, this means that there is only one $p$-adic Gibbs
measure. Here we may see some similarity with the real case.
Besides, we prove that the $p$-adic Gibbs measure is bounded, and
the generalized one is not bounded.

\end{abstract}

\maketitle

\section{introduction}

Due to the assumption that $p$-adic numbers provide a more exact and
more adequate description of microworld phenomena, starting the
1980s, various models described in the language of $p$-adic analysis
have been actively studied
\cite{ADFV},\cite{FO},\cite{MP},\cite{V1}. The well-known studies in
this area are primarily devoted to investigating quantum mechanics
models using equations of mathematical physics \cite{ADV,V2,VVZ}.
Furthermore, numerous applications of the $p$-adic analysis to
mathematical physics have been proposed in
\cite{ABK},\cite{Kh1},\cite{Kh2}. One of the first applications of
$p$-adic numbers in quantum physics appeared in the framework of
quantum logic in \cite{BC}. This model is especially interesting for
us because it could not be described by using conventional real
valued probability. Besides, it is also known \cite{Kh2,MP,Ro,VVZ}
that a number of $p$-adic models in physics cannot be described
using ordinary Kolmogorov's probability theory. New probability
models, namely $p$-adic values ones were investigated in
\cite{BD},\cite{K3},\cite{KhN}. After that in \cite{KYR} an abstract
$p$-adic probability theory was developed by means of the theory of
non-Archimedean measures \cite{Ro},\cite{kas1}. Using that measure
theory in \cite{KL},\cite{Lu} the theory of stochastic processes
with values in $p$-adic and more general non-Archimedean fields
having probability distributions with non-Archimedean values has
been developed. In particular, a non-Archimedean analog of the
Kolmogorov theorem was proven (see also \cite{GMR}). Such a result
allows us to construct wide classes of stochastic processes using
finite dimensional probability distributions\footnote{We point out
that stochastic processes on the field $\bq_p$ of $p$-adic numbers
with values of real numbers have been studied by many authors, for
example, \cite{AK,AZ1,AZ2,DF,Koc,Y}. In those investigations wide
classes of Markov processes on $\bq_p$ were constructed and studied.
In our case the situation is different, since probability measures
take their values in $\bq_p$. This leads our investigation to some
difficulties. For example, there is no information about the
compactness of $p$-adic values probability measures. }. Therefore,
this result gives us a possibility to develop the theory of
statistical mechanics in the context of the $p$-adic theory, since
it lies on the basis of the theory of probability and stochastic
processes. First steps in this theory have been started in
\cite{KM,MR1,MR2}. Note that one of the central problems of such a
theory is the study of infinite-volume Gibbs measures corresponding
to a given Hamiltonian, and a description of the set of such
measures. In most cases such an analysis depend on a specific
properties of Hamiltonian, and complete description is often a
difficult problem. This problem, in particular, relates to a phase
transition of the model (see \cite{G}).

In \cite{KK1,KK2} a notion of ultrametric Markovianity, which
describes independence of contributions to random field from
different ultrametric balls, has been introduced, and shows that
Gaussian random fields on general ultrametric spaces (which were
related with hierarchical trees), which were defined as a solution
of pseudodifferential stochastic equation (see also \cite{KaKo}),
satisfies the Markovianity. In \cite{KE,K3} Gaussian $p$-adic valued
measure was investigated, and showed that such a measure is not
bounded. This phenomena shows the difference between real and
$p$-adic valued probability theories. Some applications of the
results to replica matrices, related to general ultrametric spaces
have been investigated in \cite{KK3}.

The purpose of this paper is devoted to the development of $p$-adic
statistical mechanics in  $p$-adic probability theory framework.
Namely, we study one-dimensional countable state of nearest-neighbor
Potts models (see \cite{Ga,W}) over $p$-adic filed. We are
especially interested in the existence of phase transition for the
mentioned model. Here by the phase transition we mean the existence
of two different generalized $p$-adic Gibbs measures associated with
the model. Note that such measures present more natural concrete
examples of $p$-adic Markov processes (see \cite{KL}, for
definitions). It is worth to mention that when the number of states
of the model is finite, say $q$, then the corresponding $p$-adic
$q$-state Potts models have been studied in
\cite{MR1,MR2}\footnote{The classical (real value) counterparts of
such models were considered in \cite{W}}. It was established that a
strong phase transition occurs if $q$ is divisible by $p$. Here the
strong phase transition means the existence of two different
$p$-adic Gibbs measures. This shows that the transition depends on
the number of spins $q$ \footnote{To establish such results we
investigated $p$-adic dynamical systems associated with the model.
Note that first investigations of non-Archimedean dynamical systems
were appeared in \cite{HY} (see also
\cite{AKh,KhN,FL1,Sil1,TVW,Wo})}. Therefore, it is interesting to
know the situation in the setting with countable states. In
\cite{KMM} (see also \cite{M}) first steps to investigation of such
a countable state $p$-adic Potts model on Cayley tree have been
studied. We provided a sufficient condition for the uniqueness of
$p$-adic Gibbs measures. Note that such a condition does not depend
on the value of $p$.

In the present paper we shall consider countable state $p$-adic
Potts model on $\bz_+$. A main aim is to establish the existence of
the phase transition for the model. In our study, we essentially use
one dimensionality of the model. To show it we reduce the problem,
to the investigation of an infinite-dimensional nonlinear equation.
We will show that the derived equation has two solutions, which
yields the existence of the phase transition. Note that, in the
usual real case, if one considers one dimensional
translation-invariant model with nearest neighbor interaction, then
such a model does not exhibit a phase transition. But in our
setting, we are able to produce a model which is
translation-invariant and has nearest neighbor interactions, and for
such a model we shall prove the existence of the phase transition.
Nevertheless, we should stress that in our model there does not
occur the strong phase transition, this means that there is only one
$p$-adic Gibbs measure. Here we may see some similarity with the
real case.

Let us briefly describe the paper. After preliminaries, in section 3
we introduce the model, and define generalized $p$-adic Gibbs
measure and $p$-adic Gibbs measure, respectively. Here the provided
construction of such measures which depends on a weight $\l$.  The
goal of this investigation is to give a sufficient condition for the
existence of two such measures. Note that in comparison to a real
case, in a $p$-adic setting, \`{a} priori the existence of such kind
of measures for the model is not known, since there is not much
information on topological properties of the set of all $p$-adic
measures defined even on compact spaces. However, in the real case,
there is the so called the Dobrushin's Theorem \cite{Dob1,Dob2,G}
which gives a sufficient condition for the existence of the Gibbs
measure for a large class of Hamiltonians. Using the $p$-adic analog
of Kolmogorov's extension Theorem \cite{KL}, an investigation of the
defined measures is reduced to the examination of an
infinite-dimensional nonlinear recursion equation. In next Section
4, we associate a nonlinear operator on the Banach space $c_0$ to
the derived recursion equation. We provide a sufficient condition on
weight $\l$, which ensures the existence of two fixed points of the
nonlinear operator. This implies the existence of the phase
transition for the model. It turns out that the finding condition
does not depend on values of the prime $p$, and therefore, an
analogous fact is not true when the number of spins is finite.
Moreover, we show that the found fixed points define a $p$-adic
Gibbs and s generalized $p$-adic Gibbs measures, respectively.
Besides, we prove that the $p$-adic Gibbs measure is bounded, and
the generalized one is not bounded.

\section{Preliminaries}

Throughout the paper $p$ will be a fixed prime number greater than
3, i.e. $p\geq 3$. Every rational number $x\neq 0$ can be
represented in the form $x=p^r\dsf{n}{m}$, where $r,n\in\bz$, $m$ is
a positive integer, $(p,n)=1$, $(p,m)=1$. The $p$-adic norm of $x$
is given by
$$
|x|_p=\left\{ \ba{ll}
p^{-r} & \ \textrm{ for $x\neq 0$}\\
0 &\ \textrm{ for $x=0$}.\\
\ea \right.
$$

It satisfies the following strong triangle inequality
$$
|x+y|_p\leq\max\{|x|_p,|y|_p\},
$$
this is a non-Archimedean norm.

The completion of the field of rational numbers $\bq$  with
respect to the $p$-adic norm is called  {\it $p$-adic field} and
it is denoted by $\bq_p$.

According to non-Archimedeanity of the norm the following statement
holds true (see \cite{Ko}).

\begin{lem}\label{ser} Let $\{x_n\}$ be a sequence in $\bq_p$.
Then $\sum\limits_{k=1}^\infty x_k$ converges iff $x_n\to 0$.
\end{lem}

Note that any $p$-adic number $x\neq 0$ can be uniquely
represented in the form \begin{equation}\label{can}
x=p^{\g(x)}(x_0+x_1p+x_2p^2+\cdots),
\end{equation}
where $\g=\g(x)\in\bz$ and $x_j$ are integers, $0\leq x_j\leq p-1$,
$x_0>0$, $j=0,1,2,...$ (see more detail \cite{Ko}). In this case
$|x|_p=p^{-\g(x)}$.

We recall that an integer $a\in \bz$ is called {\it a quadratic
residue modulo $p$} if the equation $x^2\equiv a(\textrm{mod
$p$})$ has a solution $x\in \bz$.

\begin{lem}\label{sqrt}\cite{VVZ} In order that the equation
$$
x^2=a, \ \ 0\neq a=p^{\g(a)}(a_0+a_1p+...), \ \ 0\leq a_j\leq p-1,
\ a_0>0
$$
has a solution $x\in \bq_p$, it is necessary and sufficient that
the following conditions are fulfilled:
\begin{itemize}
\item[(i)] $\g(a)$ is even;\\
\item[(ii)] $a_0$ is a quadratic residue modulo $p$ if $p\neq 2$,
$a_1=a_2=0$ if $p=2$.
\end{itemize}
\end{lem}

Denote $\bz_p=\{x\in\bq_p: \ \ |x|_p\leq 1\}$. Elements of the set
$\bz_p$ are called {\it $p$-adic integers}.

\begin{lem}[Hensel's Lemma]\label{HL}\cite{Ko} Let $P(x)$ be polynomial whose the
coefficients are $p$-adic integers. Let $a_0\in\bz_p$ be a $p$-adic
integer such that $P(a_0)\equiv 0 (mod\  p)$ and $P'(a_0)\neq 0
(mod\ p)$. There exists a unique $p$-adic integer $x_0\in\bz_p$ such
that $P(x_0)=0$ and $|x_0-a_0|_p\leq 1/p$.
\end{lem}

 Given $a\in \bq_p$ and  $r>0$ put
$$
B(a,r)=\{x\in \bq_p : |x-a|_p< r\}, \ \ S(a,r)=\{x\in \bq_p :
|x-a|_p=r\}.
$$
The {\it $p$-adic logarithm} is defined by the series
$$
\log_p(x)=\log_p(1+(x-1))=\sum_{n=1}^{\infty}(-1)^{n+1}\dsf{(x-1)^n}{n},
$$
which converges for $x\in B(1,1)$; the {\it $p$-adic exponential}
is defined by
$$
\exp_p(x)=\sum_{n=0}^{\infty}\dsf{x^n}{n!},
$$
which converges for $x\in B(0,p^{-1/(p-1)})$.

\begin{lem}\label{exp}\cite{Ko} Let $x\in B(0,p^{-1/(p-1)})$
then we have
\begin{eqnarray}\label{exp}
&& |\exp_p(x)|_p=1,\ \ \ |\exp_p(x)-1|_p=|x|_p,
\ \ \ |\log_p(1+x)|_p=|x|_p \\
\label{el} && \log_p(\exp_p(x))=x, \ \ \exp_p(\log_p(1+x))=1+x.
\end{eqnarray}
\end{lem}

Note the basics of $p$-adic analysis, $p$-adic mathematical physics
are explained in \cite{Ko,R,S}.

Let $(X,\cb)$ be a measurable space, where $\cb$ is an algebra of
subsets $X$. A function $\m: \cb\to \bq_p$ is said to be a
$p$-adic measure if for any $A_1,...,A_n\subset\cb$ such that
$A_i\cap A_j=\emptyset$ ($i\neq j$) the equality holds
$$
\mu\bigg(\bigcup_{j=1}^{n} A_j\bigg)=\sum_{j=1}^{n}\mu(A_j).
$$

A $p$-adic measure is called a probability measure if $\mu(X)=1$. A
$p$-adic probability measure $\m$ is called {\it bounded} if
$\sup\{|\m(A)|_p : A\in \cb\}<\infty $. For more detail information
about $p$-adic measures we refer to \cite{Kh1},\cite{K3}.

\section{ $p$-adic Potts model and its $p$-adic Gibbs measures}

In the sequel we will use the notation $\bz_+=\{0,1,2,\cdots\}$.

Now define the $p$-adic Potts model on $\bz_+$ with spin values in
the set $\Phi=\{0,1,2,\cdots,\}$. Note that a configuration $\s$
on $\bz_+$ is defined as a function $x\in \bz_+\to\s(x)\in\Phi$;
in a similar manner one defines configurations $\s_n$ and
$\w_{(n)}$ on $[0,n]$ and $\{n\}$, respectively. The set of all
configurations on $\bz_+$ (resp. $[0,n]$, $\{n\}$) coincides with
$\Omega=\Phi^{\bz_+}$ (resp. $\Omega_{n}=\Phi^{[0,n]},\ \
\Omega_{\{n\}}=\Phi$). One can see that
$\Om_{n}=\Om_{n-1}\times\Phi$. Using this, for given
configurations $\s_{n-1}\in\Om_{n-1}$ and $\w\in\Om_{\{n\}}$ we
define their concatenations  by
\begin{equation*}
\s_{n-1}\vee\w=\bigg\{\{\s_{n-1}(k),k\in[0,n-1]\},\{\w\}\bigg\}.
\end{equation*}

The Hamiltonian $H_n:\Om_{n}\to\bq_p$ of $p$-adic countable state
Potts model has the form
\begin{equation}\label{Potts}
H_n(\s)=J\sum_{k=0}^{n-1}\delta_{\s(k),\s(k+1)},  \ \
n\in\mathbb{N},
\end{equation}
here $\s\in\Om_n$, $\delta$ is the Kronecker symbol and
\begin{equation}\label{J}
|J|_p\leq\frac{1}{p}.
\end{equation}

Note that such a condition provides the existence of a $p$-adic
Gibbs measure (see \eqref{mu}).

Let us construct $p$-adic Gibbs measures corresponding to the model.

A given set $A$ we put $\bq_p^{A}=\{\{x_i\}_{i\in A}:
x_i\in\bq_p\}.$

Assume that a function  $\h: \bn\to \bq_p^{\Phi}$, i.e.
$\h_n=\{h_{i,n}\}_{i\in\Phi}$, $n\in\bn$ is given and a non-zero
element $\l=\{\l(i)\}_{i\in\Phi}\in\bq_p^{\Phi}$ is fixed  such that
\begin{equation}\label{L}
|\l(n)|_p\to 0 \ \ \textrm{as} \ \ n\to\infty
\end{equation}
which is called a {\it weight}. In what follows, without losing
generality we may assume that $\l(0)\neq 0$.

Given $n=1,2,\dots$ a $p$-adic probability measure $\m^{(n)}_\h$ on
$\Om_{n}$ is defined by
\begin{equation}\label{mu}
\mu^{(n)}_{\h}(\s)=\frac{1}{Z^{(\h)}_{n}}\exp_p\bigg\{H_n(\s)\bigg\}h_{\s(n),n}\prod_{k=0}^n\l(\s(k)),
\end{equation}
here, $\s\in\Om_n$ and $Z_n^{(\h)}$ is the corresponding normalizing
factor called a {\it partition function} given by
\begin{equation}\label{ZN1}
Z^{(\h)}_n=\sum_{\s\in\Omega_{n}}\exp_p\bigg\{H_n(\s)\bigg\}h_{\s(n),n}\prod_{k=0}^n\l(\s(k)),
\end{equation}
here subscript $n$ and superscript $(\h)$ are accorded to the $Z$,
since it depends on $n$ and a function $\h$.

The condition \eqref{J} with non-Archimedeanity of the norm
$|\cdot|_p$ implies that $|H(\s)|_p <\frac{1}{p^{1/(p-1)}}$ for all
$\s\in\Om_n$, $n\in\bn$, which means the existence of $\exp_p$ in
\eqref{mu}, therefore, the measures $\m^{(n)}_\h$ are well defined.

One of the central results of the theory of probability concerns a
construction of an infinite volume distribution with given
finite-dimensional distributions, which is called {\it Kolmogorov's
Theorem} \cite{Sh}. Therefore, in this paper we are interested in
the same question but in a $p$-adic context. More exactly, we want
to define a $p$-adic probability measure $\m$ on $\Om$ such that it
would be compatible with defined ones $\m_\h^{(n)}$, i.e.
\begin{equation}\label{CM}
\m(\s\in\Om: \s|_{[0,n]}=\s_n)=\m^{(n)}_\h(\s_n), \ \ \ \textrm{for
all} \ \ \s_n\in\Om_{n}, \ n\in\bn.
\end{equation}

In general, \`{a} priori the existence of such a kind of measure
$\m$ is not known, since, there is not much information on
topological properties, such as compactness, of the set of all
$p$-adic measures defined even on compact spaces\footnote{In the
real case, when the state space is compact, then the existence
follows from the compactness of the set of all probability measures
(i.e. Prohorov's Theorem). When the state space is non-compact, then
there is a Dobrushin's Theorem \cite{Dob1,Dob2} which gives a
sufficient condition for the existence of the Gibbs measure for a
large class of Hamiltonians. In \cite{Ga} using that theorem it has
been established the existence of the Gibbs measure for the real
counterpart of the studied Potts model. It should be noted that
there are even nearest-neighbor models with countable state space
for which the Gibbs measure does not exists \cite{Sp}.}. Therefore,
at a moment, we can only use the $p$-adic Kolmogorov extension
Theorem (see \cite{GMR},\cite{KL}) which based on so called {\it
compatibility condition} for the measures $\m_\h^{(n)}$, $n\geq 1$,
i.e.
\begin{equation}\label{comp}
\sum_{\w\in\Phi}\m^{(n)}_\h(\s_{n-1}\vee\w)=\m^{(n-1)}_\h(\s_{n-1}),
\end{equation}
for any $\s_{n-1}\in\Om_{{n-1}}$. This condition according to the
theorem implies the existence of a unique $p$-adic measure $\m$
defined on $\Om$ with a required condition \eqref{CM}. Note that
more general theory of $p$-adic measures has been developed in
\cite{kas1,kas2}.

So, if for some function $\h$ the measures  $\m_\h^{(n)}$ satisfy
the compatibility condition, then there is a unique $p$-adic
probability measure, which we denote by $\m_\h$, since it depends on
$\h$. Such a measure $\m_\h$ is said to be {\it a generalized
$p$-adic Gibbs measure} corresponding to the $p$-adic Potts model.
By $G\cg(H)$ we denote the set of all generalized $p$-adic Gibbs
measures associated with functions $\h=\{\h_n,\ n\in\bn\}$. If
$|G\cg(H)|\geq 2$ (here $|A|$ stands for the cardinality of a set
$A$) i.e.  there are at least two different generalized $p$-adic
Gibbs measures in $G\cg(H)$, namely one can find two different
functions $\sb$ and $\h$ defined on $\bn$ such that there exist the
corresponding measures $\m_\sb$ and $\m_\h$, which are different,
then we say that {\it a phase transition} occurs for the model,
otherwise, there is {\it no phase transition}. If the function $\h$
has a special form, i.e. $\h=\{\exp_p(\kappa_{i,n})\}_{i\in\Phi}$
for some $\{\kappa_{i,n}\}\subset\bq_p$, then the corresponding
measure defined by \eqref{mu} is called {\it $p$-adic Gibbs
measure}. The set of all $p$-adic Gibbs measures is denoted by
$\cg(H)$.  If $|\cg(H)|\geq 2$, then we say that for this model
there exists {\it a strong phase transition}. Note that such kind of
measures and transitions for Ising and Potts models have been
studied in \cite{MR1,MR2,KM}.

Now one can ask for what kind of functions $\h$ the measures
$\m_\h^{(n)}$ defined by \eqref{mu} would satisfy the compatibility
condition \eqref{comp}. The following theorem gives an answer to
this question.

\begin{thm}\label{comp1}\cite{M} The measures $\m^{(n)}_\h$, $
n=1,2,\dots$ (see \eqref{mu}) satisfy the compatibility condition
\eqref{comp} if and only if for any $n\in \bn$ the following
equation holds:
\begin{equation}\label{eq1}
\hat h_{i,n}=\frac{\l(i)}{\l(0)}F_i(\hat \h_{n+1};\theta), \ \
i\in\bn
\end{equation}
here and below $\theta=\exp_p(J)$, a vector $\hat \h=\{\hat
h_i\}_{i\in\bn}\in\bq_p^\bn$ is defined by a vector
$\h=\{h_i\}_{i\in\Phi}$ as follows
\begin{equation}\label{H}
\hat h_i=\frac{h_i\l(i)}{h_0\l(0)}, \ \ \ i\in\bn
\end{equation}
and mappings $F_i:\bq_p^{\bn}\times\bq_p\to\bq_p$ are defined by
\begin{equation}\label{eq2}
F_i(\xb;\theta)=\frac{(\theta-1)x_i+\sum_{j=1}^{\infty}x_j+1}
{\sum_{j=1}^{\infty}x_j+\theta}, \ \ \xb=\{x_i\}_{i\in\bn}, \ \
i\in\bn.
\end{equation}
\end{thm}

\begin{lem}\label{inf}\cite{M} Let $\{\hat\h\}$ be a solution of \eqref{eq1} such that $\sum_{j=1}^\infty\hat h_{j,n}\neq -\t$ for every
$n\in \bn$. Then for every $n\in\bn$ one has
\begin{equation}\label{inf1}
\sum_{j=1}^\infty \hat h_{j,n}<\infty.
\end{equation}
\end{lem}

\begin{rem}\label{R2} If every sequence $\hat\h_n$ is bounded, then
\eqref{L},\eqref{H} with Lemma \ref{ser} imply that the series
$\sum\limits_{j=1}^\infty \hat h_{j,n}$ is always convergent.
\end{rem}

{\bf Observation 3.1.} Here we are going to underline a connection
between $q$-state Potts model with the defined one. First recall
that $q$-state Potts model is defined by the same Hamiltonian
\eqref{Potts}, but with the state space $\Phi_q=\{0,1,\dots,q-1\}$.
Similarly, one can define $p$-adic Gibbs measures for the $q$-state
Potts model, here instead of the weight $\{\l(i)\}$ we will take a
collection $\{\l(0),\l(1),\dots,\l(q-1)\}\subset\bq_p$.

Now consider countable Potts model with a weight $\{\l(i)\}$ such
that
\begin{equation}\label{qq}
\l(k)=0 \ \ \textrm{for all} \ k\geq q, \ q>1.
\end{equation}
In this case the corresponding Gibbs measures will coincide with
those of $q$-state Potts model. Indeed, let
\begin{eqnarray*}
&& \Om^c=\{\s\in\Om:\ \exists j\in\bz_+:\ \s(j)\geq q\}\\
&& \Om^{(q)}=\{\s\in\Om:\  \s(j)\leq q-1 \ \forall j\in\bz_+\}
\end{eqnarray*}
It is clear that $\Om^{(q)}=\Phi_q^{\bz_+}$. Let $\m$ be a Gibbs
measure of the countable Potts model with the given weight
corresponding to a solution $\h_n=\{h_{i,n}\}_{i\in\Phi}$ of
\eqref{eq1}. Note that here by Gibbs measure we mean genralized
$p$-adic Gibbs measure. From the definition \eqref{mu} we see that
the restriction of $\m$ to $\Om^c$ is zero, i.e.
$\m\lceil_{\Om^c}=0$. Moreover, from \eqref{eq1} and \eqref{qq} we
conclude that $h_{i,n}=0$ for all $i\geq q$. This means that vectors
$\h_n^{(q)}=\{h_{i,n}\}_{i\in\Phi_q}$ will be a solution of
\eqref{eq1} corresponding to the $q$-state Potts model. Therefore,
the restriction of $\m$ to $\Om^{(q)}$ coincides with Gibbs measure
of $q$-state Potts model with a weight
$\{\l(0),\l(1),\dots,\l(q-1)\}$ corresponding to a solution of
$\h_n^{(q)}$.

Hence, we conclude that under condition \eqref{qq} all Gibbs
measures corresponding to countable Potts model are described by
those measures of $q$-state Potts model.\\

Let us recall that a function $\{\h_n\}_{n\in\bn}$ is
translation-invariant if $\h_n=\h_{n+1}:=\h$ for every $n\in\bn$. It
is natural to ask is there a translation invariant solution of
\eqref{eq1}.

Now we are looking for the translation-invariant solution $\hat\h$
of \eqref{eq1}.  Then the equation can be written as follows
\begin{equation}\label{tran}
\hat h_i=\frac{\l(i)}{\l(0)}\bigg( \frac{(\theta-1)\hat
h_i+\sum_{j=1}^{\infty}\hat h_j+1} {\sum_{j=1}^{\infty}\hat
h_j+\theta} \bigg),  \ \ i\in\bn.
\end{equation}

Investigating, the derived equation \eqref{tran}, in \cite{M} we
have proved the following

\begin{thm}\label{1}\cite{M} Let $0<|J|_p<p^{-1/(p-1)}$ and
for the weight $\l$ the condition
\begin{equation}\label{l1}
\l(0)=1, \ \ \textrm{and} \ \ |\l(m)|_p<1  \qquad \forall m\in\bn.
\end{equation} be satisfied.
Then for one dimensional $p$-adic  Potts model \eqref{Potts} there
is a generalized $p$-adic Gibbs measure, i.e. $|G\cg(H)|\geq 1$.
Moreover, there is a unique $p$-adic Gibbs measure, i.e.
$|\cg(H)|=1$.
\end{thm}

\begin{rem}\label{RGG(H)} Under the condition \eqref{l1} from Theorem \ref{1} it naturally arises
a question: is it possible that $|G\cg(H)\setminus\cg(H)|\geq 1$.
It turns out this situation can occur. Indeed, let us consider
\begin{eqnarray}\label{l12}
&&\l(0)=1, \ \ \l(1)=\l(2)=a, \ \ \l(m)_p=0 \qquad \forall m\geq 3\\
\label{l12a} && |a|_p<1.
\end{eqnarray}
It is evident that in this case \eqref{l1} is satisfied. Then
\eqref{tran} reduces to
\begin{eqnarray}\label{r1}
\hat h_1&=&a\bigg( \frac{\theta\hat h_1+\hat h_2+1} {\hat
h_1+h_2+\theta} \bigg),\\
\label{r2} \hat h_2&=&a\bigg( \frac{\theta\hat h_2+\hat h_1+1}
{\hat h_1+h_2+\theta} \bigg).
\end{eqnarray}
From \eqref{r1},\eqref{r2} one gets
\begin{eqnarray}\label{r3}
\frac{\hat h_1}{\hat h_2}=\bigg(\frac{\theta\hat h_1+\hat h_2+1}
{\theta\hat h_2+\hat h_1+1} \bigg)
\end{eqnarray}
which implies
\begin{equation*}
(\hat h_1-\hat h_2)(\hat h_1+\hat h_2+1)=0.
\end{equation*}
This means that either $\hat h_1=\hat h_2$ or $\hat h_1=-\hat
h_2-1$.

Now assume that $\hat h_1=-\hat h_2-1$, and substituting it to
\eqref{r1} we immediately find that $\hat h_1=0$ and $\hat
h_2=-1$. From \eqref{H} and \eqref{l12}  one gets that $h_1=0$.
This means that associated measure \eqref{mu} is a generalized
$p$-adic Gibbs one, i.e. belongs to $G\cg(H)\setminus\cg(H)$.

Let $\hat h_1=\hat h_2$. Then again substituting it to \eqref{r1}
and after little algebra one gets
\begin{equation}\label{rq}
Q(\hat h_1)=0
\end{equation}
where $Q(x)=2x^2+((1-\t)(1-a)+1-2a)x-a$. From \eqref{l12a} one can
see that $|Q(0)|_p=|a|_p<1$ and $|Q'(0)|_p=1$. Therefore, thanks to
the Hansel Lemma the equation \eqref{rq} has two solutions $\hat
h^{(1)}, \hat h^{(2)}\in\bq_p$ such that $|\hat h^{(1)}|_p=|a|_p,$
and $|\hat h^{(2)}|_p=1$. The measure corresponding to $\hat
h^{(1)}$ due to Theorem \ref{1} belongs to $\cg(H)$. But from
\eqref{H} and \eqref{l12a} one infers that the measure associated
with $\hat h^{(2)}$ belongs to $G\cg(H)\setminus\cg(H)$.

Hence, for the model \eqref{Potts} with a weight \eqref{l12} we
have $|G\cg(H)\setminus\cg(H)|\geq 2$, since $|\cg(H)|=1$.
\end{rem}

\section{Phase transition}

In this section we are going to show that the equation \eqref{eq1}
has at least two translation-invariant solutions under some
conditions.

In this section we will assume the following
\begin{equation}\label{PT}
\l(0)=1, \ \ \l(1)=\a, \ \ \textrm{and} \ \ |\l(m)|_p<1  \qquad
\forall m\geq 2,
\end{equation}
here $\a\in\bq_p$ such that
\begin{equation}\label{PT1}
|\a|_p=1, \ \ \ |1-\a|_p\leq 1/p.
\end{equation}

It is obvious that in this case \eqref{l1} is not satisfied. Now we
are going to find translation invariant solution of \eqref{eq1},
i.e. $\hat\h_n=\hat\h_m$ for all $n,m\in\bn$. Therefore, we assume
that $\hat\h_1=(x_1,\dots,x_n,\dots)$. Let us for the sake of
shortness,
 a given sequence $\xb=\{x_j\}_{j\geq 2}$ we denote
\begin{equation}\label{x1}
X:=\sum_{j=2}^\infty x_j.
\end{equation}

If $\hat\h_1$ is a translation invariant solution, then the first
equation in \eqref{tran} with \eqref{PT} can be rewritten by
\begin{equation}\label{xx1}
x_1=\a\bigg(\frac{\t x_1+X+1}{x_1+X+\t}\bigg).
\end{equation}
We reduced the last equation to
\begin{equation}\label{xx1}
P(x_1)=0,
\end{equation}
where $P(x)=x^2+(X+\t(1-\a))x-\a(X+1)$.

Direct checking shows that
\begin{equation}\label{Hen}
P(1)=(1-\a)(X+1+\t), \ \ \ P'(1)=2+X+\t(1-\a).
\end{equation}
If $|X+2|_p=1$, then \eqref{PT1} with the Hensel's Lemma implies
that \eqref{xx1} has a solution $x_{+,1}$ belonging to $\bq_p$.
Hence, the Vieta Theorem yields that the second solution $x_{-,1}$
of \eqref{xx1} also belongs to $\bq_p$. Note that for the both
solutions $x_{\pm,1}$ due to Hensel Lemma we have
\begin{eqnarray}\label{x-1}
|x_{+,1}-1|_p\leq 1/p.
\end{eqnarray}
Now keeping in mind that $p\geq 3$, from \eqref{x-1} one finds
\begin{eqnarray}\label{x-2}
|x_{-,1}-1|_p=1.
\end{eqnarray}

In the sequel we will need an exact form of these solutions, which
can be written as follows
\begin{equation}\label{x12}
x_{\pm,1}=\frac{(\a-1)\t-X\pm\sqrt{D_X}}{2},
\end{equation}
where
\begin{eqnarray}\label{DX}
D_X&=&(X+\t(1-\a))^2+4\a(X+1) \nonumber\\
&=&\t^2(1-\a)^2+2(2X+X\t+2)(1-\a)+(X+2)^2
\end{eqnarray}

Note that the existence of the solutions $x_{\pm,1}$ yields the
existence $\sqrt{D_X}$.

Let us now substitute \eqref{x12} into $F_i$ in \eqref{eq2}, which
has a form
\begin{equation}\label{F1}
F^{(\pm)}_i(\xb;\theta)=\frac{2(\theta-1)x_i+(\a-1)\t+X\pm\sqrt{D_X}+2}
{(\a+1)\t+X\pm\sqrt{D_X}}, \ \ i\geq 2,
\end{equation}
where $\xb=\{x_i\}_{i\geq 2}$.

Note that from \eqref{H} and \eqref{L} we see that $|x_n|_p\to 0$ as
$n\to\infty$. Therefore, it is natural to consider the following
space
\begin{equation}\label{c0}
c_0=\{\{x_n\}_{n\geq 2}\subset\bq_p: \ |x_n|_p\to 0, \  \
n\to\infty\}
\end{equation}
with a norm $\|x\|=\max\limits_n|x_n|_p$. According to Lemma
\ref{ser} for any $\{x_n\}\in c_0$ we have
$\sum\limits_{j=2}^\infty x_j<\infty$.

Define
\begin{equation}\label{B}
\Bb_{r}=\{\{x_n\}\in c_0: \ \|x\|\leq r\},
\end{equation}
where $r\in\{p^k: k\in\bz\}$.  It is clear that $\Bb_{r_p}$ is a
closed subset of $c_0$. Now consider the following  mapping
\begin{equation}\label{F}
(\cf^{(\pm)}(\xb))_i=\l(i)F^{(\pm)}_i(\xb,\t), \ \ i\geq 2,
\end{equation}
where $\xb=\{x_n\}\in c_0$.

Now our aim is to show the existence of a fixed point of
$\cf^{(\pm)}$.

Put
$$
\d=\max_{i\geq 2}|\l(i)|_p.
$$
From \eqref{PT} one immediately finds that $\d<1$.

Note that according to the condition \eqref{PT} from \eqref{H} we
obtain $|x_n|_p\leq |\l(n)|_p$, $\forall n\geq 2$, which implies
that any solution of \eqref{eq1} belongs to $\Bb_{\d}$.

\begin{lem}\label{inv} Let the conditions \eqref{PT},\eqref{PT1} be satisfied for
$\l$. Then $\cf^{(+)}(\Bb_{\d})\subset \Bb_{\d}$.
\end{lem}

\begin{proof} Let $\xb\in \Bb_\d$. Then
\begin{equation}\label{est}
|X|_p=\bigg|\sum_{j=2}^\infty x_j\bigg|_p\leq\|\xb\|\leq\d.
\end{equation}
Therefore, one has $|X+2|_p=1$, and according to the above made
argument, we infer the existence of $\sqrt{D_X}$. Now using this
fact and \eqref{est}, from \eqref{DX} and \eqref{PT1} we conclude
that $\sqrt{D_X}=2+\e$ with $|\e|_p<1$, which with \eqref{PT1}
implies that
\begin{eqnarray}\label{DX2}
&&|\sqrt{D_X}+2|_p=1, \ \ \ |\a\t+\sqrt{D_X}|_p=|3|_p,\\[2mm]\label{DX22}
&&|\sqrt{D_X}-2|_p\leq \frac{1}{p}, \ \ \ |\a\t-\sqrt{D_X}|_p=1.
\end{eqnarray}

Then by means of \eqref{est},\eqref{DX2} and \eqref{PT1} we have
\begin{eqnarray*}
|(\cf^{(+)}(\xb))_i|_p&=&|\l(i)|_p\bigg|\frac{2(\theta-1)x_i+(\a-1)\t+X+\sqrt{D_X}+2}
{(\a+1)\t+X+\sqrt{D_X}}
\bigg|_p\\
&=&|\l(i)|_p\bigg|\frac{\sqrt{D_X}+2}{\a\t-1+\t-1+\sqrt{D_X}+2}\bigg|_p\\
&=&|\l(i)|_p\leq \d
\end{eqnarray*}
for all $i\geq 2$, which implies $\cf^{(+)}(\Bb_{\d})\subset
\Bb_{\d}$. This completes the proof.
\end{proof}

Before going to the main result we need some auxiliary facts.

\begin{lem}\label{DXY} One has
$$
\sqrt{D_X}-\sqrt{D_Y}=\frac{(X-Y)(X+Y+2\t(1-\a)+4\a)}{\sqrt{D_X}+\sqrt{D_Y}}.
$$
\end{lem}
\begin{proof} From \eqref{DX} we immediately find
$$
D_X-D_Y=(X-Y)(X+Y+2\t(1-\a)+4\a)
$$
which with
$$
\sqrt{D_X}-\sqrt{D_Y}=\frac{D_X-D_Y}{\sqrt{D_X}+\sqrt{D_Y}}
$$
implies the assertion.
\end{proof}

Denote
\begin{equation}\label{xi}
\xi_X=(\a-1)\t+X+\sqrt{D_X}+2.
\end{equation}

From the direct calculation we can prove the following

\begin{lem}\label{xiDXY} One has
\begin{eqnarray}\label{xi1}
&&\xi_X-\xi_Y=X-Y+\sqrt{D_X}-\sqrt{D_Y};\\\label{xi2}
&&Y\xi_X-X\xi_Y=((\a-1)\t+2)(Y-X)+Y\sqrt{D_X}-X\sqrt{D_Y};\\
\label{xi3}
&&\xi_X\sqrt{D_Y}-\xi_Y\sqrt{D_X}=((\a-1)\t+2)(\sqrt{D_Y}-\sqrt{D_X})+X\sqrt{D_Y}-Y\sqrt{D_X}.
\end{eqnarray}
\end{lem}

Now we are in a pose to formulate the main estimation.

\begin{thm}\label{main} Let the conditions \eqref{PT},\eqref{PT1} be satisfied for
$\l$. Then one has
\begin{eqnarray}\label{fxy+}
\|\cf^{(+)}(\xb)-\cf^{(+)}(\yb)\|\leq \d |\t-1|_p\|\xb-\yb\|,
\end{eqnarray}
for every $\xb,\yb\in \Bb_{\d}$.
\end{thm}

\begin{proof} Let $\xb,\yb\in\Bb_{\d}$, then  from \eqref{F} we have
\begin{eqnarray}\label{FXY1}
|\cf^{(+)}(\xb))_i-\cf^{(+)}(\yb)_i|_p&=&|\l(i)|_p\bigg|
\frac{2(\t-1)x_i+\xi_X}{(\a+1)\t+X+\sqrt{D_X}}-\frac{2(\t-1)y_i+\xi_Y}{(\a+1)\t+Y+\sqrt{D_Y}}\bigg|_p\nonumber\\
&=&|\l(i)|_p\bigg|2(\t-1)\bigg[\underbrace{x_i((\a+1)\t+Y+\sqrt{D_Y})-y_i((\a+1)\t+X+\sqrt{D_X})}_{\textrm{I}}\bigg]+\nonumber\\
&&\underbrace{((\a+1)\t+Y+\sqrt{D_Y})\xi_X-((\a+1)\t+X+\sqrt{D_X})\xi_Y}_{\textrm{II}}\bigg|_p.
\end{eqnarray}

Now step by step, let us estimate I and II.

Put
\begin{eqnarray}\label{Del}
\D=\frac{X+Y+2\t(1-\a)+4\a}{\sqrt{D_X}+\sqrt{D_Y}}.
\end{eqnarray}

Let us first consider I. Then using Lemma \ref{DXY} one finds
\begin{eqnarray}\label{I}
\textrm{I}&=&(x_i-y_i)[(\a+1)\t+X+\sqrt{D_X}]+x_i[Y-X+\sqrt{D_Y}-\sqrt{D_X}]\nonumber
\\[2mm]
&=&(x_i-y_i)[(\a+1)\t+X+\sqrt{D_X}]-x_i(1+\D)(X-Y)
\end{eqnarray}

Now turn to II. We easily find that
\begin{eqnarray*}
\textrm{II}=(\a+1)\t(\xi_X-\xi_Y)+Y\xi_X-X\xi_Y+\xi_X\sqrt{D_Y}-\xi_Y\sqrt{D_X}
\end{eqnarray*}
according to Lemmas \ref{xiDXY} and \ref{DXY} one gets
\begin{eqnarray}\label{II}
\textrm{II}&=&2(1-\t)(Y-X)+2(1-\t)(\sqrt{D_Y}-\sqrt{D_X})\nonumber
\\[2mm]
&=&2(1-\t)(1+\D)(Y-X)
\end{eqnarray}

Then substituting \eqref{I} and \eqref{II} with \eqref{Del} into
\eqref{FXY1} we obtain
\begin{eqnarray}\label{FXY2}
|\cf^{(+)}(\xb))_i-\cf^{(+)}(\yb)_i|_p&=&|\t-1|_p|\l(i)|_p\bigg|
(x_i-y_i)((\a+1)\t+X+\sqrt{D_X})\nonumber\\
&&+(1-x_i)(1+\D)(X-Y)\bigg|_p\nonumber \\
&\leq&
|\t-1|_p|\l(i)|_p\max\bigg\{|x_i-y_i|_p,|1+\D|_p|X-Y|_p\}\nonumber\\
&\leq&|\t-1|_p|\l(i)|_p\max_{i\geq 2}\{|x_i-y_i|\}\nonumber \\
&\leq& \d|\t-1|_p\|\xb-\yb\|,
\end{eqnarray}
here we have used \eqref{x1} and $|\D|_p=1$.

Consequently, from \eqref{FXY2} we get the required inequality.
\end{proof}

Now let us turn to $\cf^{(-)}$. This case is a little bit tricky.
Therefore, impose some extra conditions. Namely, we assume
\begin{eqnarray}\label{PT2}
&& |\a-1|_p\leq\frac{1}{p^2},\\\label{PT3} &&
|\t-1|_p=\frac{1}{p}.
\end{eqnarray}

\begin{lem}\label{inv1} Let the conditions \eqref{PT},\eqref{PT1},\eqref{PT2},\eqref{PT3} be satisfied.
Then $\cf^{(-)}(\Bb_{p^{-1}\d})\subset \Bb_{p^{-1}\d}$.
\end{lem}

\begin{proof} Let $\xb\in \Bb_{p^{-1}\d}$. Then from \eqref{PT},\eqref{est} we have
$|X|_p\leq 1/p^2$. Using this with \eqref{PT2}, from \eqref{DX} one
finds that $\sqrt{D_X}=2+\e_1$, where $|\e_1|_p\leq 1/p^2$. This
with \eqref{x12} yields that
\begin{equation}\label{dx2}
|x_{-,1}+1|_p\leq \frac{1}{p^2}.
\end{equation}
Whence with \eqref{PT3} one finds
\begin{equation}\label{domin}
|x_{-,1}+X+\t|_p=|x_{-,1}+1+X+\t-1|_p=|\t-1|_p.
\end{equation}
Consequently, using \eqref{domin},\eqref{dx2} and \eqref{PT3} we
have
\begin{eqnarray*}
|(\cf^{(-)}(\xb))_i|_p&=&|\l(i)|_p\bigg|\frac{(\theta-1)x_i+x_{-,1}+X+1}
{x_{-,1}+X+\t}\bigg|_p\\
&\leq &\frac{1}{p}|\l(i)|_p\leq p^{-1}\d
\end{eqnarray*}
for all $i\geq 2$, which implies the assertion. \end{proof}

By the same argument of the proof of Theorem \ref{main} one can
prove

\begin{thm}\label{main1} Let the conditions \eqref{PT},\eqref{PT1},\eqref{PT2},\eqref{PT3} be satisfied.
Then one has
\begin{eqnarray}\label{fxy-}
\|\cf^{(-)}(\xb)-\cf^{(-)}(\yb)\|\leq \d\|\xb-\yb\|,
\end{eqnarray}
for every $\xb,\yb\in \Bb_{p^{-1}\d}$.
\end{thm}

Now we are ready to formulate our main result.

\begin{thm}\label{phase} Let the conditions
\eqref{PT},\eqref{PT1},\eqref{PT2},\eqref{PT3} be satisfied. Then a
phase transition occurs for the countable state $p$-adic Potts model
\eqref{Potts}.
\end{thm}

\begin{proof} From the conditions
\eqref{PT},\eqref{PT1},\eqref{PT2},\eqref{PT3} we infer that
statements of both Theorems \ref{main} and \ref{main1} are valid.
Noting that $\d<1$ with Theorem \ref{main} ( resp. Theorem
\ref{main1}) we can apply the fixed point theorem to $\cf^{(+)}$
(resp. $\cf^{(-)}$), which means that the existence of a unique
fixed point $\xb_+=\{x_{+,i}\}\in\Bb_{\d}$ (resp.
$\xb_-=\{x_{-,i}\}\in\Bb_{p^{-1}\d}$). Hence, equation \eqref{eq1}
has at least two translation-invariant solutions $(x_{+,1},\xb_+)$
and $(x_{-,1},\xb_-)$. These solutions according to Theorem
\ref{comp1} define $\m_+$ and $\m_-$ generalized $p$-adic Gibbs
measures, respectively. To show that such measures are different, it
is enough to establish that $\xb_+$ and $\xb_-$ are different.
Therefore, using \eqref{domin} one finds
\begin{eqnarray*}
|x_{+,i}-x_{-,i}|_p&=&|(\cf^{(+)}(\xb))_i-(\cf^{(-)}(\xb))_i|_p\\
&=&|\l(i)|_p\bigg|\frac{(\t-1)x_{+,i}+x_{+,1}+X+1}{x_{+,1}+X+\t}-
\frac{(\t-1)x_{-,i}+x_{-,1}+X+1}{x_{-,1}+X+\t}\bigg|_p\\
&=&|\l(i)|_p\frac{|\t-1|_p|x_{+,1}-x_{-,1}|_p|x_i-1|_p}{|x_{+,1}+X+\t|_p|x_{-,1}+X+\t|_p }\\
&=&|\l(i)|_p\frac{|\t-1|_p|x_{+,1}-x_{-,1}|_p}{|\t-1|_p }\\
&=&|\l(i)|_p|x_{+,1}-x_{-,1}|_p.
\end{eqnarray*}
From $|x_{+,1}-x_{-,1}|_p=|\sqrt{D_X}|_p=1$ we conclude that
$\|\xb_+-\xb_-\|=\d$, which means the measures $\m_+$ and $\m_-$ are
different.
\end{proof}

\begin{rem}  Note that if the conditions
\eqref{PT},\eqref{PT1} are  not satisfied, then it may exist only
one generalized translation-invariant $p$-adic Gibbs measure.
Indeed, consider weights defined by \eqref{l12} with
\begin{eqnarray}\label{ra}
&& |a|_p=1, \  \ \ \ |2a-1|_p\leq 1/p,\\
\label{rsq} && \sqrt{a} \ \ \textrm{does not exist in} \ \ \bq_p.
\end{eqnarray}
From Remark \ref{RGG(H)} we already knew that \eqref{tran} has a
solution $\hat h_1=0$ and $\hat h_2=-1$ which defines a
generalized translation-invariant $p$-adic Gibbs measure.

Now we show that equation \eqref{rq} does not have any solution
belonging to $\bq_p$. Indeed, one can compute that its discriminant
has a form
\begin{eqnarray*}
D=((1-\t)(1-a)+1-2a)^2+8a
\end{eqnarray*}
due to $|1-\t|_p\leq 1/p$ and \eqref{ra} one finds that
$|D|_p=|a|_p$. Hence, the assumption \eqref{rsq} with Lemma
\ref{sqrt} implies that $\sqrt{D}$ does not exist in $\bq_p$. This
means there is no solution of \eqref{rq} belonging to $\bq_p$.\\
\end{rem}

So, we have two different generalized $p$-adic Gibbs measures. It is
natural to ask: which of them would be a $p$-adic Gibbs measure?

Now recall that a translation-invariant generalized $p$-adic
measure associated with $\h=\{h_i\}\in\bq_p^\Phi$ would be
$p$-adic Gibbs one, if there is a sequence $\{\k_i\}\in\bq_p^\Phi$
such that the equality $h_i=\exp_p\k_i$ holds for all $i\in\Phi$.

Let us find the corresponding sequence $\{\k_i\}$ for
$(x_{+,1},\xb_+)$. From \eqref{H} we have
\begin{equation*}
\exp_p(\k_i-\k_0)\a=x_{+,i}, \ \ \ i\in\bn
\end{equation*}

Since $(x_{\pm,1},\xb_\pm)$ is a fixed point of \eqref{eq1},
therefore from \eqref{xx1} one gets
\begin{equation}\label{hx1}
\exp_p(\k_1-\k_0)=\frac{\t x_{+,1}+X_++1}{x_{+,1}+X_++\t}
\end{equation}
and
\begin{equation}\label{hx2}
\exp_p(\k_i-\k_0)=\frac{(\t-1)x_{+,i}+x_{+,1}+X_++1}{x_{+,1}+X_++\t},
\ \ \ i\geq 2,
\end{equation}
where as before
\begin{equation}\label{x+-}
X_\pm=\sum_{j=2}^\infty x_{\pm,j}.
\end{equation}
 By means of \eqref{x-1} and Lemma \ref{inv} one gets
\begin{eqnarray*}
&& \bigg|\frac{\t x_{+,1}+X_++1}{x_{+,1}+X_++\t}\bigg|_p=1, \\
&&\bigg|\frac{\t
x_{+,1}+X+1}{x_{+,1}+X+\t}-1\bigg|_p=|\t-1|_p|x_{+,1}-1|_p<
1/p,\\
&&\bigg|\frac{(\t-1)x_{+,i}+x_{+,1}+X_++1}{x_{+,1}+X_++\t}\bigg|_p=1,\\
&&\bigg|\frac{(\t-1)x_{+,i}+x_{+,1}+X_++1}{x_{+,1}+X_++\t}-1\bigg|_p=|\t-1|_p|x_{+,i}-1|_p=|\t-1|_p\leq
1/p,
\end{eqnarray*}
which, thanks to Lemma \ref{exp}, allow us to take $\log_p$ from
both sides of \eqref{hx1} and \eqref{hx2}. Hence putting $\k_0=0$,
we able to find $\{\k_i\}$. This means that $\m_+$ is a $p$-adic
Gibbs measure.

Now turn to $(x_{-,1},\xb_-)$. From \eqref{PT3},\eqref{domin} and
\eqref{x-2} we have
\begin{eqnarray*}
\bigg|\frac{(\t-1)x_{-,i}+x_{-,1}+X_-+1}{x_{-,1}+X_-+\t}\bigg|_p\leq
1/p.
\end{eqnarray*}
This due to Lemma \ref{exp} implies that $\xb_-$ can not be
represented as $\exp_p\k_i$. Therefore, $\m_-$ is a strictly
generalized $p$-adic Gibbs measure.

So, we have

\begin{thm}\label{phase1} Assume all the conditions of Theorem \ref{phase} are satisfied.
Then $\m_+$ is a $p$-adic Gibbs measure, but $\m_-$ is a
generalized $p$-adic Gibbs measure.
\end{thm}

In Theorem \ref{1} we have provided a sufficient condition on
uniqueness of the $p$-adic Gibbs measure. But nevertheless, it is
interesting to know whether the measure $\m_+$ is a unique $p$-adic
Gibbs measure.

\begin{thm}\label{uni4} Let the conditions \eqref{PT},\eqref{PT1} be satisfied for
$\l$. Then the measure $\m_+$ is a unique $p$-adic Gibbs measure for
$p$-adic Potts model \eqref{Potts}, i.e. $|\cg(H)|=1$.
\end{thm}

\begin{proof} Let $\gh=\{\hat
h_{1,n},\hat\h_{n}\}$ be any solution of \eqref{eq1}, where
$\hat\h_{n}=\{\hat h_{i,n}\}_{i\geq 2}$ and $\hat
h_{i,n}=\l(i)\exp(\k_{i,n})$, $i\in\bn$. We will show that such a
solution coincides with $(x_{+,1},\xb_+)$. Indeed, it is clear that
$|\hat h_{1,n}|_p=1$ and $\|\hat \h_n\|=\d$.

Let us fix $n\in\bn$ and consider the difference
\begin{eqnarray}\label{uni1}
|\hat h_{1,n}-x_{+,1}|_p&=&|a|\bigg| \frac{\t \hat
h_{1,n+1}+H_{n+1}+1}{\hat h_{1,n+1}+H_{n+1}+\t}- \frac{\t
x_{+,1}+X_++1}{x_{+,1}+X_++\t}\bigg|_p\nonumber \\
&=&|\t-1|_p|(\hat h_{1,n+1}-1)(X_+-H_{n+1})+(\hat h_{1,n+1}-x_{+,1})(\t+1+H_{n+1})|_p\nonumber \\
&\leq&|\t-1|_p\max\{|\hat h_{1,n+1}-x_{+,1}|_p,\|\hat
\h_{n+1}-\xb_+\|\},
\end{eqnarray}
here we have used that
$$
H_{n+1}:=\sum_{k=2}^\infty \hat h_{k,n+1}
$$
and $|H_{n+1}|_p\leq\|\hat\h_{n+1}\|<1$.

Similarly reasoning as in the proof of Theorem \ref{main} one gets
\begin{eqnarray}\label{uni2}
|\hat h_{i,n}-x_{+,i}|_p&=&|\l(i)|_p|F_i(\hat h_{1,n+1},\hat
\h_{n+1};\t)-F_i(x_{+,1},\xb_{+};\t)|_p\nonumber\\
&\leq &|\l(i)|_p|\t-1|_p\max\{|\hat h_{1,n+1}-x_{+,1}|_p,\|\hat
\h_{n+1}-\xb_+\|\}
\end{eqnarray}
From \eqref{uni1},\eqref{uni2} we obtain
\begin{eqnarray}\label{uni3}
\max\{|\hat h_{1,n}-x_{+,1}|_p,\|\hat \h_{n}-\xb_+\|\}\leq|\t-1|_p
\max\{|\hat h_{1,n+1}-x_{+,1}|_p,\|\hat \h_{n+1}-\xb_+\|\}
\end{eqnarray}

Now take an arbitrary $\e>0$ and $n_0\in\bn$ such that
$|\t-1|_p^{n_0}<\e$. Then iterating \eqref{uni3} $n_0$ times, one
gets
\begin{eqnarray*}
\max\{|\hat h_{1,n}-x_{+,1}|_p,\|\hat
\h_{n}-\xb_+\|\}\leq|\t-1|_p^{n_0}<\e
\end{eqnarray*}
Due to the arbitrariness of $\e$ we have $h_{1,n}=x_{+,1}$ and
$\hat\h_n=\xb_+$ for every $n\in \bn$.
\end{proof}

\begin{rem}\label{Rl1} The proved Theorem \ref{uni4} indicates that the condition \eqref{l1}
was a sufficient for the uniqueness of the $p$-adic Gibbs measure.
\end{rem}

Now let us turn to $\m_-$. Take any solution $\gh_-$ of \eqref{eq1}
of the form $\gh_-=\{x_{-,1},\hat\h^{(-)}_{n}\}$), where
$\hat\h_n^{(-)}\in\Bb_{p^{-1}\d}$. As a consequence of Theorem
\ref{main1} one can formulate the following

\begin{cor}\label{uniq} Assume that the conditions of Theorem \ref{phase} are satisfied.
Then $\gh_-$ coincides with $(x_{-,1},\xb_-)$.
\end{cor}

\begin{proof} Now fix any vertex $n\in\bn$ and $i\geq 2$. From \eqref{fxy-}  one gets
\begin{eqnarray}\label{hxx}
|\hat h^{(-)}_{i,n}-x_{-,i}|_p&=&|\l(i)|_p\bigg|F^{(-)}_i(\hat\h^{(-)}_{n+1},\t)-F^{(-)}_i(\xb_-,\t)\bigg|_p\nonumber\\
&\leq& \d \|\hat\h^{(-)}_{n+1}-\xb_-\|
\end{eqnarray}

Take an arbitrary $\e>0$ and $n_0\in\bn$ such that $\d^{n_0}<\e$.
Then \eqref{hxx} implies that
\begin{eqnarray*}
\|\hat\h^{(-)}_n-\xb_-\|&\leq&
\d\|\hat\h^{(-)}_{n+1}-\xb_-\|\leq\cdots \leq
\d^{n_0}\|\hat\h^{(-)}_{n+n_0}-\xb_-\|<\e.
\end{eqnarray*}
Hence, from the arbitrariness of $\e$ we obtain
$\hat\h^{(\pm)}_n=\xb_-$ for every $n\in \bn$.
 This proves the assertion.
\end{proof}

Now let us consider more concrete examples.

\begin{ex}\label{ex2} Assume that  $\{\hat h_m\}$ is a
solution of \eqref{tran} defined by
\begin{equation}\label{sol}
\hat h_m=p^{m+1}-\frac{J^{m-1}}{(m-1)!}, \ \ \ m\in\bn,
\end{equation}
where $|J|_p=1/p$. Then one can see
$$
\sum_{m=1}^\infty \hat h_m=\frac{p^2}{1-p}-\t,
$$
here as before $\t=\exp_p(J)$. Now substituting \eqref{sol} to
\eqref{tran} we obtain the corresponding $\l$ by
\begin{equation}\label{ex0}
\frac{\l(m)}{\l(0)}=\frac{J^{m-1}}{(m-1)!}\bigg(\frac{p^2}{(p-1)(\t-1)(\frac{J^{m-1}}{(m-1)!}-1)+p^2}\bigg).
\end{equation}
Put $\l(0)=1$. Then it is clear that $\l(1)=1$. The equality
$|\t-1|_p=|J|_p=1/p$ implies that
\begin{eqnarray}\label{ex01}
\bigg|\frac{p^2}{(p-1)(\t-1)(\frac{J^{m-1}}{(m-1)!}-1)+p^2}\bigg|_p
=\frac{1}{p^2|\t-1|_p}=\frac{1}{p}.
\end{eqnarray}
This with $|J^{m-1}/(m-1)!|_p<1$ yields that the conditions of
Theorem \ref{phase} are satisfied. Now keeping in mind that $\hat
h_m$ is a solution of \eqref{tran} and by means of
\eqref{H},\eqref{ex01} we find that
\begin{eqnarray}\label{ex}
\bigg|\frac{h_m}{h_0}\bigg|_p&=&\bigg|\frac{(p-1)(\t-1)(\frac{J^{m-1}}{(m-1)!}-1)+p^2}{p^2}\bigg|_p=p
\end{eqnarray}
for every $m\geq 2$. Hence, \eqref{ex} shows that
$|\frac{h_m}{h_0}|_p>1$ which means that the equality
$h_m=\exp_p(\k_m)$ impossible for any $\k_m$. This implies that
the corresponding measure belongs to $G\cg(H)\setminus\cg(H)$.
From Corollary \ref{uniq} we infer that the constructed
generalized $p$-adic Gibbs measure is $\m_-$. Thanks to Theorem
\ref{uni4} for the weights \eqref{ex0} there is also a unique
$p$-adic Gibbs measure.
\end{ex}

\begin{ex}\label{ex3} Now suppose that  $\{\hat h_m\}$ is a solution of \eqref{tran}
defined by $\hat h_m=\frac{J^{m-1}}{(m-1)!}$, $m\in\bn$, with
$|J|_p\leq 1/p$. Then $ \sum_{m=1}^\infty \hat h_m=\t$. By the
same argument used in Example \ref{ex2} one can define $\l$, for
which the conditions \eqref{PT},\eqref{PT1} are satisfied as well.
In this case one can show that $|\frac{h_m}{h_0}|_p=1$. Hence,
according to Theorem \ref{uni4} the corresponding measure is a
unique $p$-adic Gibbs measure.
\end{ex}

Now let us turn to the boundedness of the measures $\m_+$ and
$\m_-$. We need the following

\begin{lem} Let $\gh$ be a translation-invariant solution of \eqref{eq1}, and
$\m_\gh$ be the associated Gibbs measure. Then for the
corresponding partition function $Z^{(\gh)}_n$ (see \eqref{ZN1})
the following equality holds
\begin{equation}\label{ZN2}
Z^{(\gh)}_{n+1}=A_{\gh}Z^{(\gh)}_n,
\end{equation}
where
\begin{equation}\label{aN}
A_{\gh}=\l(0)\bigg(\t+\sum_{j=1}^\infty\hat h_j\bigg).
\end{equation}
\end{lem}

\begin{proof} From \eqref{eq1} we conclude that there is a constant
$A_{\gh}\in\bq_p$ such that
\begin{equation}\label{aN1}
\sum_{j\in\Phi}\exp_p\{J\d_{ij}\}h_{j}\l(j)=A_{\gh}h_i
\end{equation}
for any $i\in\Phi$.

On the other hand, using \eqref{mu} and \eqref{aN1} we have
\begin{eqnarray*}
1&=&\sum_{\s\in\Om_n}\sum_{\w\in\Phi}\m^{(n+1)}_\gh(\s\vee\w)\\
&=&\sum_{\s\in\Om_n}\sum_{\w\in\Phi}\frac{1}{Z^{(\gh)}_n}\exp_p\{H(\s\vee\w)\}h_\w\prod_{k=0}^{n}\l(\s(k))\l(\w)\\
&=&\sum_{\s\in\Om_n}\frac{1}{Z^{(\gh)}_{n+1}}\exp_p\{H(\s)\}\prod_{k=0}^{n}\l(\s(k))
\sum_{j\in\Phi}\exp_p\{J\d_{\s(n),j}\}h_j\l(j)\\
&=&\frac{A_{\gh}}{Z^{(\gh)}_{n+1}}\sum_{\s\in\Om_n}\exp_p\{H(\s)\}h_{\s(n)}\prod_{k=0}^{n}\l(\s(k))\\
&=&\frac{A_{\gh}}{Z^{(\gh)}_{n+1}}Z^{(\gh)}_n
\end{eqnarray*}
which implies \eqref{ZN2}. From \eqref{aN1} we may easily find
\eqref{aN}.
\end{proof}

Now we are ready to formulate a result.

\begin{thm}\label{phase2} Assume all the conditions of Theorem \ref{phase} are satisfied. Then
the $p$-adic Gibbs measure $\m_+$ is bounded, but the generalized
$p$-adic Gibbs measure  $\m_-$ is not bounded.
\end{thm}

\begin{proof} Let us first consider $\m_+$. Take any $\s\in\Om_n$.
Then from \eqref{mu} with \eqref{ZN2}, \eqref{aN} one gets
\begin{eqnarray*}
|\m_+(\s)|_p&=&\frac{1}{|Z^{(+)}_n|_p}\bigg|\exp_p\{H(\s)\}x_{+,\s(n)}\prod_{k=0}^{n-1}\l(\s(k))\bigg|_p\\
&=&\frac{1}{|(x_{+,1}+X_++\t)^{n-1}Z^{(+)}_1|_p}|x_{+,\s(n)}|_p\prod_{k=0}^{n-1}|\l(\s(k))|_p\\
&\leq &\frac{1}{|Z^{(+)}_1|_p},
\end{eqnarray*}
this means that $\m_+$ is bounded.

Now consider $\m_-$. Let us take
$$
\s^{(1)}_n=\{\s(k)=1, k\in [0,n]\}.
$$
Then analogously as above with
\eqref{PT}, \eqref{domin} and \eqref{PT3} we find
\begin{eqnarray*}
|\m_-(\s^{(1)}_n)|_p&=&\frac{1}{|Z^{(-)}_n|_p}\bigg|\exp_p\{H(\s^{(1)}_n)\}x_{-,1}\prod_{k=0}^{n-1}\l(1)\bigg|_p\\
&=&\frac{|x_{-,1}\a^n|_p}{|(x_{-,1}+X_-+\t)^{n-1}Z^{(-)}_1|_p}\\
&=&\frac{1}{|\t-1|^{n-1}|Z^{(-)}_1|_p}\\
&=&\frac{p^{n-1}}{|Z^{(-)}_1|_p},
\end{eqnarray*}
which means that $\m_-$ is not bounded.
\end{proof}

\begin{rem}
In \cite{MR1} we have proved that at $p=3$ there is two $p$-adic
Gibbs measures for that $3$-state Potts model, i.e. $|G(H)|\geq 2$.
Hence a strong phase transition occurs. There, it was shown that
those $p$-adic Gibbs measures were unbounded. Hence, Theorem
\ref{phase2} shows the difference between the finite state Potts
models, since there the $p$-adic Gibbs measures are unbounded when
the phase transition occurs.
\end{rem}

\section*{Acknowledgement}  The present study have been done within
the grant FRGS0409-109 of Malaysian Ministry of Higher Education. A
part of this work was done at the Abdus Salam International Center
for Theoretical Physics (ICTP), Trieste, Italy. The author thanks
the ICTP for providing financial support during his visit as a
Junior Associate at the centre.



\begin{thebibliography}{99}


\bibitem{AK} Albeverio S., Karwowski W. A random walk on $p$-adics, the
generator and its spectrum, \textit{ Stochastic. Process. Appl.}
{\bf 53} (1994) 1-Ц22.

\bibitem{AZ1} Albeverio S.,  Zhao X. On the relation between different
constructions of random walks on $p$-adics, \textit{ Markov Process.
Related Fields} {\bf 6} (2000) 239-Ц256.

\bibitem{AZ2} Albeverio S., Zhao X. Measure-valued branching processes
associated with random walks on $p$-adics, \textit{Ann. Probab.}
{\bf 28}(2000) 1680-Ц1710.



\bibitem{AKh} Anashin V., Khrennikov A., \textit{Applied Algebraic Dynamics},
Walter de Gruyter, Berlin, New York, 2009.


\bibitem{ADV}  Aref\'eva I. Ya., Dragovic B.,
Volovich I.V. $p-$ adic summability of the anharmonic ocillator,
\textit{Phys. Lett. B} {\bf 200}(1988) 512--514.

\bibitem{ADFV} Aref\'eva I. Ya., Dragovic B., Frampton P.H.,  Volovich I.V.
The wave function of the Universe and $p -$ adic gravity,
\textit{Int. J. Modern Phys. A} {\bf 6}(1991) 4341--4358.

\bibitem{AV1} Arrowsmith D.K., Vivaldi F., Some $p-$adic
representations of the Smale horseshoe, \textit{Phys. Lett. A} {\bf
176}(1993), 292--294.


\bibitem{ABK}  Avetisov V.A., Bikulov A.H., Kozyrev S.V. Application of
pЦadic analysis to models of spontaneous breaking of the replica
symmetry, \textit{J. Phys. A: Math. Gen.} {\bf 32}(1999) 8785--8791.

\bibitem{BC} Beltrametti E., Cassinelli G. Quantum mechanics and $p-$ adic
numbers, \textit{Found. Phys.} {\bf 2}(1972) 1--7.


\bibitem{BD} Besser A., Deninger C., $p$-adic Mahler measures, \textit{J. Reine
Angew. Math.} {\bf 517} (1999), 19Ц50.

\bibitem{B}  Baxter R.J. {\it Exactly  Solved Models in Statistical Mechanics},
(Academic Press, London/New York, 1982).


\bibitem{DF} Del Muto M., Fig$\grave{a}$-Talamanca A. Diffusion on locally
compact ultrametric spaces, \textit{Expo. Math.} {\bf 22}(2004)
197Ц-211.


\bibitem{Dob1}  Dobrushin R.L. The problem of uniqueness of a
Gibbsian random field and the problem of phase transitions, \textit{
Funct.Anal. Appl.} {\bf 2} (1968) 302--312.

\bibitem{Dob2} Dobrushin R.L. Prescribing a system of random variables by conditional
distributions, \textit{Theor. Probab. Appl.} {\bf 15}(1970)
458--486.


\bibitem{FL1}  Fan A. H., Li M.T., Yao J.Y., Zhou D., Strict ergodicity of
affine $p$-adic dynamical systems on $Z_p$, \textit{Adv. Math.},
{\bf 214} (2007), 666-Ц700.


\bibitem{FO}   Freund P.G.O., Olson M. Non-Archimedian
strings, \textit{Phys. Lett. B} {\bf 199}(1987) 186--190.

\bibitem{Ga}  Ganikhodjaev N.N. The Potts model on $\bz^d$  with countable set of spin
values, \textit{Jour. Math. Phys.} {\bf 45}(2004), 1121--1127.

\bibitem{GMR} Ganikhodjaev N.N., Mukhamedov F.M., Rozikov U.A.
Phase transitions of the Ising model on $\bz$ in the $p$-adic number
field, \textit{Uzbek. Math. Jour.} {\bf 4} (1998) 23--29 (Russian).


\bibitem{G} Georgii H.O. \textit{ Gibbs measures and phase transitions},
Walter de Gruyter, Berlin, 1988.

\bibitem{HY} Herman M., Yoccoz J.-C.,  Generalizations of some
theorems of small divisors to non-Archimedean fields, In: Geometric
Dynamics (Rio de Janeiro, 1981), Lec. Notes in Math. 1007, Springer,
Berlin, 1983, pp.408-447.

\bibitem{KaKo} Kaneko H., Kochubei A.N.,  Weak solutions of
stochastic differential equations over the field of $p$-adic
numbers, \textit{Tohoku Math. J.} {\bf 59}(2007), 547--564.


\bibitem{KM} Khamraev M., Mukhamedov F.M. On $p$-adic $\lambda$-model on
the Cayley tree, \textit{Jour. Math. Phys.} {\bf 45}(2004)
4025--4034.

\bibitem{kas1} Katsaras A.K.  Extensions of $p$-adic vector measures,
\textit{Indag. Math.N.S.} {\bf 19} (2008) 579--600.

\bibitem{kas2} Katsaras A.K.  On spaces of $p$-adic vector measures,
    \textit{P-Adic Numbers, Ultrametric Analysis, Appl.} {\bf  1} (2009)  190--203.

\bibitem{K3} Khrennikov A.Yu.  $p$-adic valued probability measures, \textit{Indag. Mathem. N.S.}
{\bf 7}(1996) 311--330.


\bibitem{Kh1} Khrennikov A.Yu. \textit{ $p$-adic Valued Distributions in Mathematical
Physics}, Kluwer Academic Publisher, Dordrecht, 1994.

\bibitem{Kh2} Khrennikov A.Yu.  \textit{Non-Archimedean analysis: quantum paradoxes,
dynamical systems and biological models}, Kluwer Academic Publisher,
Dordrecht, 1997.

\bibitem{KE} Khrennikov A.Yu., Endo M., Unboundedness of a $p$-adic Gausian distribution,
\textit{Russ. Acad. Sci., Izv., Math.} {\bf 41}(1993), 367-375.

\bibitem{KK1} Khrennikov A.Yu., Kozyrev S.V., Wavelets on ultrametric spaces,
\textit{Appl. Comput. Harmonic Anal.}, {\bf 19}(2005) 61-76.

\bibitem{KK2} Khrennikov A.Yu., Kozyrev S.V., Ultrametric random
field, \textit{Infin. Dimens. Anal. Quantum Probab. Relat. Top.}
{\bf 9}(2006), 199-213.

\bibitem{KK3} Khrennikov A.Yu., Kozyrev S.V.,  Replica symmetry breaking
related to a general ultrametric space I,II,III, \textit{Physica A},
{\bf 359}(2006), 222-240; 241-266; {\bf 378}(2007), 283-298.

\bibitem{KL} Khrennikov A.Yu., Ludkovsky S.  Stochastic processes on
non-Archimedean spaces with values in non-Archimedean fields,
\textit{Markov Process. Related Fields}
 {\bf 9}(2003) 131--162.

\bibitem{KMM} Khrennikov A., Mukhamedov F., Mendes J.F.F. On $p$-adic Gibbs measures of countable
state Potts model on the Cayley tree, \textit{Nonlinearity} {\bf
20}(2007) 2923Ц-2937.

\bibitem{KhN} Khrennikov  A.Yu., Nilsson M.  \textit{$p$-adic deterministic and
random dynamical systems}, Kluwer, Dordreht, 2004.

\bibitem{KYR} Khrennikov  A.Yu.,  Yamada S., van Rooij A., Measure-theoretical
approach to $p$-adic probability theory, \textit{Annals Math. Blaise
Pascal} {\bf 6} (1999) 21--32.

\bibitem{Ko} Koblitz N.  \textit{$p$-adic numbers, $p$-adic analysis and
zeta-function}, Berlin, Springer, 1977.

\bibitem{Koc} Kochubei A.N. \textit{Pseudo-differential equations and stochastics over non-Archimedean
fields}, Mongr. Textbooks Pure Appl. Math. 244 Marcel Dekker, New
York, 2001.




\bibitem{Lu}  Ludkovsky S.V.  Non-Archimedean valued quasi-invariant
descending at infinity measures, \textit{ Int. J. Math. Math. Sci.}
{\bf 2005}(2005) N. 23, 3799--3817.



\bibitem{MP} Marinary E., Parisi G.  On the $p$-adic five point function, \textit{ Phys. Lett. B}\
 {\bf 203}(1988) 52--56.



\bibitem{M} Mukhamedov F.M., On the existence of generalized Gibbs
measures for the one-dimensional $p$-adic countable state Potts
model, \textit{Proc. Steklov Inst. Math.} {\bf 265} (2009) 165-Ц176.


\bibitem{MR1} Mukhamedov F.M., Rozikov U.A. On Gibbs measures of $p$-adic
Potts model on the Cayley tree, \textit{Indag. Math. N.S.} {\bf 15}
(2004) 85--100.

\bibitem{MR2}  Mukhamedov F.M., Rozikov U.A. On inhomogeneous $p$-adic Potts
model on a Cayley tree, \textit{Infin. Dimens. Anal. Quantum Probab.
Relat. Top.} {\bf 8}(2005) 277--290.


\bibitem{R} A.M.Robert, \textit{A course of $p$-adic analysis}, Springer, New York, 2000.

\bibitem{Ro} van Rooij A., \textit{Non-archimedean functional analysis}, Marcel Dekker, New York, 1978.

\bibitem{S}  Schikhof W.H. \textit{Ultrametric Calculus}, Cambridge
University Press, Cambridge, 1984.


\bibitem{Sil1} Silverman J.H.  \textit{The arithmetic of dynamical
systems}. Graduate Texts in Mathematics 241, New York, Springer,
2007.


\bibitem{Sh} Shiryaev A.N. \textit{Probability}, Nauka, Moscow, 1980.

\bibitem{Sp} Spitzer F. Phase transition in one-dimensional
nearest-neighbor systems, \textit{J. Funct. Anal.} {\bf 20}(1975),
240--255.

\bibitem{TVW} Thiran E., Verstegen D., Weters J., $p$-adic
dynamics, \textit{J.Stat. Phys.} {\bf 54}(3/4)(1989), 893--913.

\bibitem{VVZ} Vladimirov V.S., Volovich I.V., Zelenov E.I. \textit{ $p$-adic Analysis and
Mathematical Physics}, World Scientific, Singapour, 1994.

\bibitem{V1} Volovich I.V. Number theory as the ultimate physical theory,
\textit{p-Adic Numbers, Ultrametric Analysis Appl.} {\bf 2}(2010), 77-Ц87;// Preprint TH.4781/87, 1987.

\bibitem{V2}  Volovich I.V. $p-$adic string, \textit{Classical Quantum Gravity} {\bf 4} (1987) L83-L87.

\bibitem{Y} Yasuda K., Extension of measures to infinite-dimensional spaces
over $p$-adic field, \textit{Osaka J. Math.} {\bf 37}(2000)
967-Ц985.

\bibitem{Wo} Woodcock C.F.,  Smart N.P., $p$-adic chaos and random number
generation, \textit{Experiment Math.} {\bf 7} (1998) 333-Ц342.

\bibitem{W} Wu F.Y., The Potts model, \textit{Rev. Mod. Phys.} {\bf 54} (1982) 235--268.



\end{thebibliography}
\end{document}